\title{Reachability by Paths of Bounded Curvature in a Convex Polygon
  \thanks{O.C.~was supported by Mid-career Researcher Program
    through NRF grant funded by the~MEST (No.~R01-2008-000-11607-0).}}

\documentclass[a4paper]{article}

\usepackage{fullpage}
\usepackage{graphicx}

\usepackage{latexsym,amsmath,amssymb,amsthm}

\newcommand{\DD}{{\cal D}}

\newcommand{\vecs}{{\mathbf{s}}}
\newcommand{\dir}{{\mathbf{d}}}

\newcommand{\vecr}{{\mathbf{r}}}
\newcommand{\vech}{{\mathbf{h}}}

\newcommand{\vecq}{{\mathbf{q}}}

\newcommand{\RR}{{\mathbb{R}}}

\newcommand{\R}[1]{\RR^{#1}}

\newcommand{\ldisk}{{D_L}}
\newcommand{\rdisk}{{D_R}}

\newcommand{\bd}{\partial}

\newcommand{\thetastar}{\theta^{\ast}}
\newcommand{\before}{\preccurlyeq}

\DeclareMathOperator{\reach}{\textsc{reach}}
\DeclareMathOperator{\FIL}{\textsc{fil}}
\DeclareMathOperator{\BFIL}{\textsc{bfil}}
\DeclareMathOperator{\LFIL}{\textsc{lfil}}
\DeclareMathOperator{\fil}{Fil}
\DeclareMathOperator{\DA}{\textsc{da}}
\DeclareMathOperator{\LDA}{\textsc{lda}}
\DeclareMathOperator{\RDA}{\textsc{rda}}

\DeclareMathOperator{\FC}{\textsc{fc}}
\DeclareMathOperator{\conv}{\textsc{conv}}

\let\leq\leqslant
\let\geq\geqslant

\usepackage{hyperref}

\newtheorem{theorem}{Theorem}
\newtheorem{lemma}[theorem]{Lemma}
\newtheorem{prop}[theorem]{Proposition}
\newtheorem{observation}[theorem]{Observation}

\newcommand{\epsfigure}[2]{
  \begin{figure}[htb]
    \centerline{\includegraphics{bcr-#1}}
    \caption{#2}
    \label{f:#1}
  \end{figure}}

\makeatletter

\partopsep\z@ \textfloatsep 10pt plus 1pt minus 4pt
\def\section{\@startsection {section}{1}{\z@}{-3.5ex plus -1ex minus
-.2ex}{2.3ex plus .2ex}{\large\bf}}
\def\subsection{\@startsection{subsection}{2}{\z@}{-3.25ex plus -1ex
minus -.2ex}{1.5ex plus .2ex}{\normalsize\bf}}
\def\@fnsymbol#1{\ensuremath{\ifcase#1\or *\or 1\or 2\or
    3\or 4\or 5\or 6\or 7 \or 8\ or 9 \or 10\or 11 \else\@ctrerr\fi}}

\makeatother

\author{Hee-Kap Ahn
  \thanks{Department of Computer
    Science and Engineering, Pohang University of Science and Technology,
    San 31, Hyoja-dong, Nam-gu, Pohang, Korea. Email: heekap@postech.ac.kr.}
  \and
  Otfried Cheong
  \thanks{Department of Computer Science, KAIST, Gwahangno~335,
    Yuseong-gu, Daejeon, Korea. Email: otfried@kaist.edu.}
  \and 
  Ji\v{r}\'{\i} Matou\v{s}ek
  \thanks{Dept.~of Applied Mathematics and Institute of Theoretical
    Computer Science (ITI), Charles University,
    Malostransk\'{e} n\'{a}m. 25, 118~00~~Praha~1, Czech Republic.
  Email: matousek@kam.mff.cuni.cz.}  
  \and 
  Antoine Vigneron
  \thanks{INRA, UR341 Math\'ematiques et Informatique Appliqu\'ees, 
  Domaine de Vilvert, F-78352 Jouy-en-Josas cedex. 
  Email: antoine.vigneron@jouy.inra.fr.}}

\begin{document}

\maketitle

\begin{abstract}
  Let $B$ be a point robot moving in the plane, whose path is
  constrained to forward motions with curvature at most one, and let
  $P$ be a convex polygon with $n$ vertices.  Given a starting
  \emph{configuration} (a location and a direction of travel) for $B$
  inside $P$, we characterize the region of all points of $P$ that can
  be reached by~$B$, and show that it has complexity~$O(n)$.  We give
  an $O(n^2)$ time algorithm to compute this region.  We show that a
  point is reachable only if it can be reached by a path of type
  CCSCS, where C denotes a unit circle arc and S denotes a line
  segment.
\end{abstract}

\section{Introduction}

The problem of planning the motion of a robot subject to non-holonomic
constraints~\cite{l-rmp-91, ss-ampr-90} (for instance, bounds on
velocity or acceleration~\cite{crr-eakpp-91, dxcr-kmp-93,
  rs-mppmo-94}, bounds on the turning angle) has received
considerable attention in the robotics literature.  Theoretical
studies of non-holonomic motion planning are far sparser.

In this paper we consider a point robot in the plane whose turning
radius is constrained to be at least one and that is not allowed to
make reversals.  This restriction corresponds naturally to constraints
imposed by the steering mechanism found in car-like robots.  We assume
that the robot is located at a given position (and orientation) inside
a convex polygon, and we are interested in the set of points in the
polygon that can be reached by the robot.  We put no restriction on
the orientation with which the robot can reach a point.

The lack of such a restriction distinguishes our work from most of the
previous theoretical work on curvature-constrained paths, which
usually assumes that not a point, but a \emph{configuration} (a
location with orientation) is given.  Dubins~\cite{d-cmlca-57} was
perhaps the first to study curvature-constrained shortest paths.  He
proved that a curvature-constrained shortest path from a given
starting configuration to a given final configuration consists of at
most three segments, each of which is either a straight line or an arc of
a unit-radius circle.  Reeds and Shepp~\cite{rs-opctg-90} extended
this characterization to robots that are allowed to make reversals.
Using ideas from control theory, Boissonnat et al.~\cite{bcl-spbcp-94}
gave an alternative proof for both cases, and Sussmann~\cite{s-sppcb-95}
extended the characterization to the 3-dimensional case.  

In the presence of obstacles, Fortune and Wilfong~\cite{fw-pcm-91}
gave a single-exponential decision procedure to verify if two given
configurations can be joined by a curvature-constrained path avoiding
the polygonal obstacles.  On the other hand, computing a
\emph{shortest} bounded-curvature path among polygonal obstacles is
NP-hard, as shown by Reif and Wang~\cite{rw-ctdcc-98}.
Wilfong~\cite{w-mpav-88} designed an exact algorithm for the case
where the curvature-constrained path is limited to some fixed straight
``lanes'' and circular arc turns between the lanes.  Agarwal et
al.~\cite{art-mpscr-95} considered the case of disjoint convex
obstacles whose curvature is also bounded by one, and gave efficient
approximation algorithms. Boissonnat and Lazard~\cite{bl-ptacs-03}
gave a polynomial-time algorithm for computing the exact shortest
paths for the case when the edges of the obstacles are circular arcs
of unit radius and straight line segments. Boissonnat et
al.~\cite{bgkl-accsp-02} gave an $O(n^4)$ algorithm for finding a
\emph{convex and simple} path of bounded curvature within a
\emph{simple polygon}.  Agarwal et al.~\cite{ablrsw-ccspc-02}
presented an $O(n^2\log n)$-time algorithm to compute a
curvature-constrained shortest path between two given configurations
inside a \emph{convex polygon}.  They also showed that there exists an
optimal path that consists of at most eight line segments or circular
arcs.  For general polygonal obstacles, Backer and
Kirkpatrick~\cite{bk-caasb-08} recently gave the first complete
approximation algorithm, improving on earlier work that approximated
the shortest ``robust'' path~\cite{jc-pspmr-92,aw-aaccs-01}.

At least two interesting problems have been studied where not
configurations but only locations for the robot are given.  The first
problem considers a \emph{sequence} of points in the plane, and asks
for the shortest curvature-constrained path that visits the points in
this sequence.  In the second problem, the \emph{Dubins traveling
  salesman problem}, the input is a \emph{set} of points in the plane,
and asks to find a shortest curvature-constrained path visiting all
points.  Both problems have been studied by researchers in the
robotics community, giving heuristics and experimental
results~\cite{sfb-opptspdv-05,mc-rhpdtsp-06,leny_07}. From a
theoretical perspective, Lee et al.~\cite{lcksc-accsp-00} gave a
linear-time, constant-factor approximation algorithm for the first
problem.  No approximation algorithms are known for the Dubins
traveling salesman problem.

Our result is a characterization of the region of points reachable by
paths under curvature constraints from a given starting configuration
inside a convex polygon~$P$.  We show that all points reachable from
the starting configuration are also reachable by paths of type CCSCS,
where C denotes an arc of a unit-radius circle and S denotes a line
segment. When $P$ has $n$ vertices, we show that the reachable region
has complexity $O(n)$, and we give an $O(n^2)$ time algorithm to
compute this region.

\section{Terminology and some lemmas}

Let $P$ be a convex polygon in the plane.  A \emph{configuration}
$\vecs = (s, \dir)$ is a point $s$ together with a direction of travel
$\dir$ (a unit vector).  By a \emph{path}, we mean a continuously
differentiable curve (the image of a $C^1$-mapping of $[0,1]$ to
$\R2$) with average curvature bounded by one in every positive-length
interval. Unless stated otherwise, we assume that a path is
completely contained in~$P$.  A \emph{configuration on the path $\pi$}
is a configuration $\vecs = (s,\dir)$ with $s$ on $\pi$ such that
$\dir$ is the forward tangent to $\pi$ in $s$.  The \emph{starting
  configuration} of $\pi$ is the starting point of $\pi$ with its
forward tangent.  A \emph{simple path} is a path with no
self-intersection; we allow the endpoints of a simple path to
coincide, in which case we call it a \emph{simple closed
  path}. (Hence, a simple closed path is smooth except possibly at one
point.)

Given a configuration $\vecs = (s, \dir)$, the \emph{left disk}
$\ldisk(\vecs)$ (\emph{right disk} $\rdisk(\vecs)$) is the unit disk
touching $s$ and completely contained in the left (right) halfplane
defined by the directed line through $\dir$. (All \emph{unit disks} in
this paper are \emph{unit-radius} disks.)

The \emph{left directly accessible region} $\LDA(\vecs)$ is the set of
all points in $P$ that can be reached by a path with starting
configuration $\vecs$ consisting of a single (possible zero-length)
circular arc on the boundary of $\ldisk(\vecs)$ followed by a single
(possible zero-length) line segment.  The \emph{right directly
  accessible region} is defined analogously.  The \emph{directly
  accessible region} $\DA(\vecs)$ is the union of the left and right
directly accessible region.

\paragraph{Pestov-Ionin lemma.}

The following lemma is perhaps the foundation for all our results.  In
a slightly less general form, it was proven by Pestov and
Ionin~\cite{pi-lpcigcc-59}.
Recently, it has been used for a curve 
reconstruction problem~\cite{gt-rcdc-05}. 
\begin{lemma}[Pestov-Ionin]\label{l:gPI}
  Any simple closed path contains a unit disk in its interior.
\end{lemma}
For sake of completeness, we sketch a proof analogous to
Pestov and Ionin's.
\begin{lemma}\label{l:ggPI}
  Let $D$ be a closed disk, and $\Gamma$ be a simple 
  path with endpoints $(a,b)$ such that $\Gamma \cap D=\{a,b\}$. 
  Then there is a unit disk touching~$\Gamma \setminus \{a,b\}$ that
  lies within the region $\cal R$ bounded by $\Gamma$ and the exterior
  arc of $\bd D$ 
\epsfigure{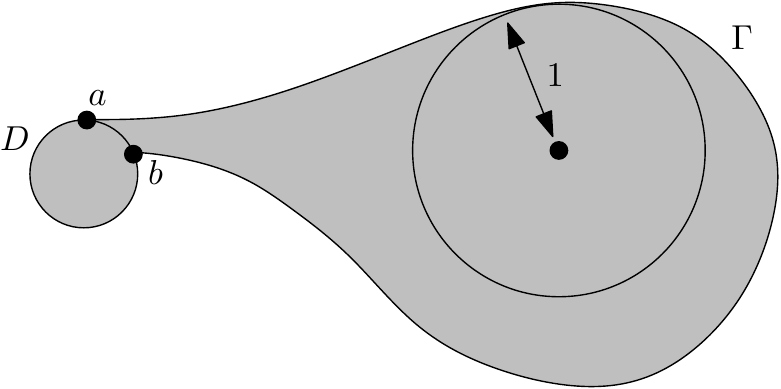}{Illustration of Lemma~\ref{l:ggPI}. The 
  region $\mathcal{R}$ is shaded.}
  (See Figure~\ref{f:pestov}).
\end{lemma}
\begin{proof}
  We proceed by induction on the length $\ell$ of $\Gamma$ (or, more
  precisely, by induction on $\lfloor \ell/\pi \rfloor$).
  
  When $\ell < \pi$, we can prove by integration that a unit disk
  tangent to $\Gamma$ does not cross $\Gamma$. We consider a unit disk
  $D_0$ tangent to $\Gamma$ at $m \notin \{a,b\}$ on the interior
  side.  Since $D_0 \setminus D$ has only one connected component and
  $m \in D_0 \setminus D$, clearly $D_{0}$ is contained in~$\cal R$.
  
  Otherwise, let $m$ be the point halfway between $a$ and $b$ on
  $\Gamma$. Let $D_1$ be the largest disk contained in $\cal R$ that
  is tangent to $\Gamma$ in $m$.  If the radius of $D_1$ is larger or
  equal to $1$, then we are done.  If not, $D_{1}$ must touch
  $\bd {\cal R}$ in another point~$m'$.  Clearly $m'$ lies on
  $\Gamma$, and the length of the arc $\Gamma'$ of $\Gamma$ between
  $m$ and $m'$ is at most $\ell/2$.  By the induction hypothesis, a
  unit disk $D_{2}$ lies inside the region ${\cal R}_{1}$ bounded by
  $\Gamma'$ and $D_{1}$. Since ${\cal R}_{1} \subseteq {\cal R}$, the
  lemma follows.
\end{proof}
Lemma~\ref{l:gPI} follows from Lemma~\ref{l:ggPI} by observing that it
still holds when $D$ degenerates to a point.

\paragraph{Filling.}

By $\FIL(P)$ we denote the \emph{set} of all unit-radius disks that
are completely contained in $P$, and we let $\fil(P)$ be the union of
all the disks in $\FIL(P)$. (See Figure~\ref{f:elp}.)  Both $\FIL(P)$
and $\fil(P)$ will be called the \emph{filling} of $P$.
\epsfigure{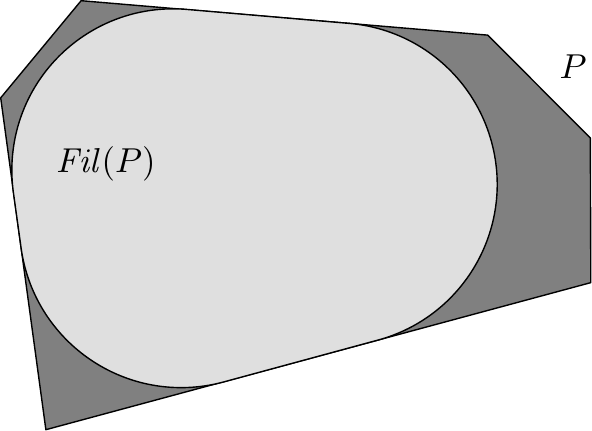}{The filling $\fil(P)$ is shaded in light grey, and the 
three pockets are shaded in dark grey.}

\paragraph{Pockets.}

The connected components of $P\setminus \fil(P)$ are called the
\emph{pockets} of $P$.  A pocket of $P$ is bounded by a single
circular arc (lying on one disk of $\FIL(P)$) and a connected part of
the boundary of $P$. The first and last edge on this connected chain
are called the \emph{mouth edges} of the pocket.  Its extremities are
called the \emph{mouth points}.  The mouth edges form an angle smaller
than~$\pi$ (this is equivalent to observing that the mouth points lie
on the same disk of $\FIL(P)$ and form an angle smaller than
$\pi$)~\cite{ablrsw-ccspc-02}. Agarwal et al.~\cite{ablrsw-ccspc-02}
proved the following lemma.

\begin{lemma}[Pocket lemma]\label{l:pocket}
  A path entering a pocket from $\fil(P)$ cannot leave the pocket
  anymore.
\end{lemma}
\begin{proof}
  We consider a pocket $K$ bounded by a disk $D$. Suppose that there
  is a path $\Gamma$ that enters and leaves $K$.  There is a subpath
  $\Gamma'$ of $\Gamma$ whose endpoints lie on $D$ and whose other
  points are in the interior of $K$.  By Lemma~\ref{l:ggPI}, there is
  a unit disk $D'$ that touches $\Gamma'$ at a point in the interior
  of $K$, and is contained in $K \cup D$. Hence, $D' \in \FIL(P)$, and
  $D'$ intersects the interior of $K$, a contradiction.
\end{proof}

\paragraph{Reachability for a union of disks.}
 
For a set $\DD$ of unit disks, we let $\conv(\DD)$ denote the set of
all unit disks contained in the convex hull of $\bigcup \DD$.
Equivalently, $\conv(\DD)$ consists of the unit disks centered at
points of the convex hull of the set of all centers of the disks in
$\DD$.
\begin{observation}\label{obs:convd}
  \label{l:conv}
  Let $\DD$ be a set of unit disks in the plane. Then $\bigcap \DD =
  \bigcap\conv(\DD)$.
\end{observation}

Given a convex set $Q$, a \emph{configuration on the boundary} of $Q$
is a configuration $\vecs = (s,\dir)$ with $s$ on the boundary of $Q$
and such that $\dir$ is tangent to the boundary of $Q$ in~$s$.
\begin{lemma}\label{l:nonreach}
  Let $\DD$ be a set of unit disks, and let $\vecs$ be a starting
  configuration on the boundary of $\bigcup\conv(\DD)$. Then no point in
  the interior of $\bigcap \DD$ can be reached by a path starting at
  $\vecs$ and contained in $\bigcup\conv(\DD)$, or even in any convex
  polygon $P$ such that $\FIL(P)=\conv(\DD)$.
\end{lemma}
\begin{proof}
  Assume to the contrary that there is a path $\gamma$ with starting
  configuration $\vecs$ on the boundary of $\bigcup\conv(\DD)$ and
  ending point~$t$ in the interior of $\bigcap \DD$.  Assume for the
  moment that $\gamma$ lies completely in $\bigcup\conv(\DD)$.  
  
  We extend $\gamma$ to infinity using a straight ray, such that the
  extended path is still $C^{1}$.  We then extend $\gamma$ backwards,
  by attaching a single loop around the boundary of
  $\bigcup\conv(\DD)$ at $\vecs$. To summarize, the extended path
  $\gamma'$ starts at $\vecs$, makes a single loop around the boundary
  of $\bigcup\conv(\DD)$, then follows the original path $\gamma$, and
  finally escapes to infinity along a straight line.  We can now
  construct a simple, closed path $\gamma''$ as follows: starting at
  infinity, we follow $\gamma'$ \emph{backwards}, until we encounter
  the first intersection of $\gamma'$ with the part that we have
  already seen.  
  \epsfigure{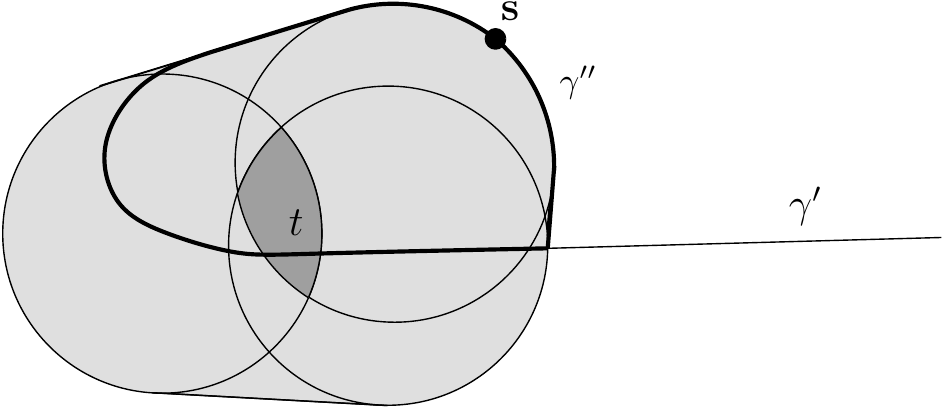}{Proof of Lemma~\ref{l:nonreach}. The set $\DD$
    consists of three disks. The convex hull $\conv(\DD)$ is shaded in
    light grey, and $\bigcap \DD$ is in dark grey.  The path $\gamma$ from
    $s$ to $t$ is enlarged into the path $\gamma''$.}
  Such an intersection must exist since the two
  extensions intersect.  We define~$\gamma''$ to be the part of~$\gamma'$
  between these two self-intersection points.
  Observe that~$t$ lies either on or outside the closed loop~$\gamma''$.
  
  By the Pestov-Ionin Lemma, $\gamma''$ contains a unit-disk $D$. 
  Since $\gamma''$ is contained in $\bigcup\conv(\DD)$, we have $D \in
  \conv(\DD)$. Consequently, the interior of $\bigcap \DD$ lies in the
  interior of~$D$, a contradiction with the fact that~$t$ must lie
  on or outside~$\gamma''$.

  The lemma still holds when we allow the path to lie inside a convex
  polygon~$P$ with $\FIL(P) = \conv(\DD)$.  After all, by
  Lemma~\ref{l:pocket}, the path cannot return to $\fil(P)$ after it
  has left it.
\end{proof}

\paragraph{The characterization.}
Consider a starting configuration $\vecs$ on the boundary of the
filling $\fil(P)$.  It is easy to see that any point in $P$ not in
$\bigcap \FIL(P)$ can be reached by a path from $\vecs$.  On the other
hand, by Lemma~\ref{l:nonreach}, no point in $\bigcap \FIL(P)$
can be reached, and so we have a complete characterization of the
region reachable from $\vecs$ as the complement of $\bigcap \FIL(P)$.

If the starting configuration lies on the boundary of an arbitrary
unit disk contained in~$P$, the same characterization holds.  For
arbitrary starting configurations, however, the situation becomes far
more complicated.  There is the possibility that no path starting at $\vecs$
is tangent to the boundary of $\fil(P)$,
and the filling has no relation to the reachable region.

If there exists a path starting at $\vecs$ that is tangent to the
boundary, all points outside $\bigcap \FIL(P)$ are reachable, but it
is still possible that some points inside $\bigcap \FIL(P)$ are
reachable, for instance because they lie in the directly reachable
area $\DA(\vecs) = \LDA(\vecs)\cup\RDA(\vecs)$, or because
$\vecs$ lies in a pocket with additional maneuvering space (that we
would not have been able to exploit if starting inside $\fil(P)$ by
Lemma~\ref{l:pocket}). (See Figure~\ref{f:man}).
\epsfigure{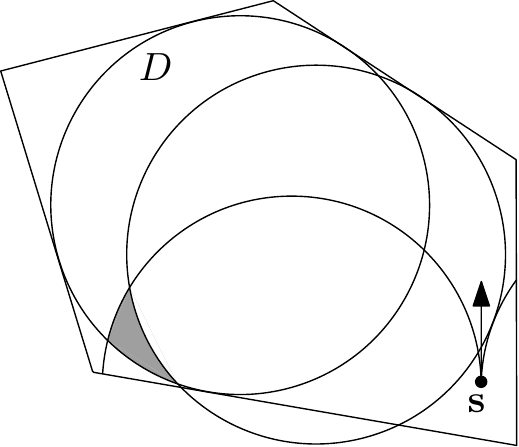}{Points in the grey region are reachable from $\vecs$,
but are neither in $\DA(\vecs)$ nor in the complement of $D = \bigcap
\FIL(P)$.}

In the rest of this paper, we give a complete characterization of the
reachable region, for any starting configuration in~$P$.  Let us
denote the set of points $t \in P$ such that $t$ is reachable by a
path starting from a configuration~$\vecs$ by $\reach(\vecs)$.

\section{Paths starting along the boundary}
\label{s:side}

In this section, we assume that the starting configuration $\vecs$ is 
on the boundary
of $P$ (recall that this means also that the direction is tangent to
the boundary).  Without loss of generality, we also assume that the
direction of $\vecs$ is counterclockwise along the boundary, so that
points of $P$ are reached locally by a left turn from $\vecs$.  It
turns out that in this situation we can restrict ourselves to paths
containing no right-turning arcs.

The \emph{forward chain} $\FC(\vecs)$ is the longest subchain of the
boundary of $P$, that starts counterclockwise from $\vecs$, and 
that turns by an angle at most $\pi$. (See Figure~\ref{f:fchain}.)
\epsfigure{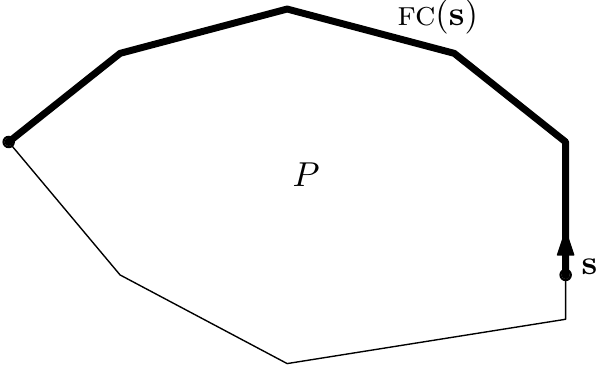}{The forward chain $\FC(\vecs)$.}
In other words, when $\vecs$
is directed vertically upward, this
chain contains all the edges of $P$ that are above its
interior, as well as the part of the edge that contains $s$ and
is above $s$, and, if there is one, the other vertical edge of
$P$.

If the forward chain intersects the interior of $\ldisk(\vecs)$, then we 
have the following simple description of the reachable region.
\begin{lemma}\label{lem:pocketreach}
  If the forward chain $\FC(\vecs)$ intersects the interior of $\ldisk(\vecs)$,
  then $\reach(\vecs) = \LDA(\vecs)$.
\end{lemma}
\begin{proof}
  The left directly accessible region $\LDA(\vecs)$ can be enlarged 
  to a pocket of the
  left disk $\ldisk(\vecs)$. Thus, no point outside $\LDA(\vecs)$ is
  reachable by the Pocket Lemma (Lemma~\ref{l:pocket}).
\end{proof}
 
We denote by $\LFIL(\vecs)$ the set of disks contained in $P \cup
\ldisk(\vecs)$ that touch the forward chain.  Note that $\LFIL(\vecs)$
always contains the left disk $\ldisk(\vecs)$.  Let us remark that the
set of centers of the disks in $\LFIL(\vecs)$ need not be connected.
For example, Figure~\ref{f:nonc} shows a situation where
$\LFIL(\vecs)$ consists of just three disks; their centers are marked by
black dots (in general, any number of connected components is
possible).
\epsfigure{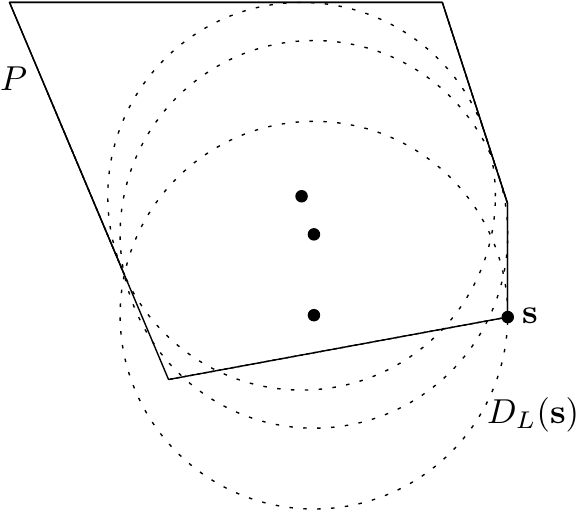}{Example where the left filling consists of just 3
  disks.} 
For a disk $F \in \LFIL(\vecs)$, we define $\LDA(F)$ as
$\LDA(\vecr)$, where $\vecr$ is the first configuration on $\FC(\vecs)$
touching~$F$ (by the above, this is well defined).  Note
that if $F \in \FIL(P)$, then $\LDA(F)$ is simply the complement of the
interior of~$F$.
We continue with a lemma on the reachable
points outside $\LDA(\vecs)$.
\begin{lemma}\label{l:tinDL}
  Let $P$ and $\vecs$ be as above, with $\ldisk(\vecs)\notin
  \FIL(P)$, and the forward chain $\FC(\vecs)$ not intersecting the
  interior of $\ldisk(\vecs)$.  
  Suppose that there is a path $\gamma$ from $\vecs$ to
  $t$, where $t \notin \LDA(\vecs)$. Then there exists a disk
  $F$ such that $F \in \FIL(P)$ and $t$ is not in the interior of $F$, 
  or there
  exists a disk $F \in \LFIL(\vecs)$  such that $t \in \LDA(F)$.
\end{lemma}
\begin{proof}
  Without loss of generality, we assume that $\vecs$ is directed
  vertically upward. We first assume that $t$ is in the
  interior of $\ldisk(\vecs)$.
 
  Consider the line~$st$. Let us first assume that $\gamma$ intersects
  this line top-to-bottom (or tangentially) at~$t$. In that case, we
  extend $\gamma$ forward by the semi-infinite ray starting at $t$,
  and backward by the boundary of~$\ldisk(\vecs)$.  We trace the
  resulting path backwards 
  from infinity, stopping at the first intersection of the path with
  the part we have already seen, and thus forming a loop that does not
  contain~$t$ in its interior.  We apply the Pestov-Ionin lemma to
  this loop, and find a unit-radius disk $D$ contained in it. (See
  Figure~\ref{f:proof1}.)
  \epsfigure{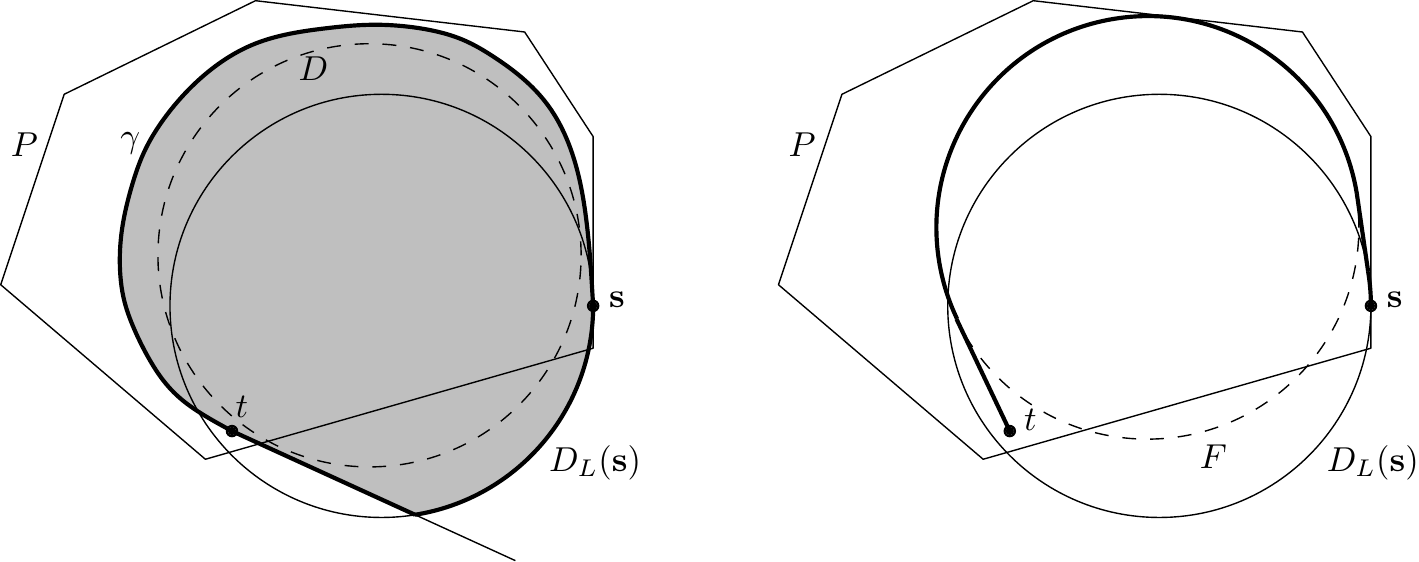}{(left) A unit disk $D$ lies in the shaded area.
  (right) $t$ is in $\LDA(F)$.}
  If $D \in \FIL(P)$, then we are done. Otherwise, note that the loop
  does not cross the boundary of $P \cup \ldisk(\vecs)$, so
  $D \in P \cup \ldisk(\vecs)$. Then we obtain a disk $F \in
  \LFIL(\vecs)$ by translating $D$ upwards until it touches the
  forward chain $\FC(\vecs)$,  and we have $t \in \LDA(F)$.

  Now we consider the case where $\gamma$ intersects the line~$st$
  bottom-to-top at~$t$. Let $t'$ be the first point of intersection of
  $\gamma$ and the line~$st$ (along the path~$\gamma$).  If $t'$ lies
  between $s$ and $t$ on the line~$st$, then we can apply the argument
  above to conclude the existence of a disk $F$ such that $t' \in
  \LDA(F)$---but then also $t \in \LDA(F)$.

  If, finally, $t$ lies between $s$ and $t'$ on~$st$, then we extend
  $\gamma$ by the semi-infinite ray starting in $t$.  This ray must
  intersect the part of $\gamma$ from $s$ to $t'$, and so again we
  have found a loop lying in~$P$, fulfilling the requirements of the
  Pestov-Ionin lemma, and containing~$t$.  As above, there is then a
  disk $F\in \FIL(P)$ with $t\in \LDA(F)$.

  We now consider the case where $t$ does not lie in the interior of 
  $\ldisk(\vecs)$.  Then $t$ lies in a connected component $P_c$ of 
  $P \setminus \ldisk(\vecs)$ 
  different from $\LDA(\vecs)$. Since $\ldisk(\vecs)$ does not intersect the
  forward chain, $P_c$ lies entirely below $\ldisk(\vecs)$. Let $t'$ denote the
  first point on $\gamma$ that is on the boundary of $P_c$. We trace
  backward from $t'$ a path along $\gamma$, and then along the lower semi-circle
  of $\ldisk(\vecs)$, until we reach a point that we have already seen. It forms
  a loop on which we apply the Pestov-Ionin lemma. Thus we find a
  disk $D$ inside $P \cup \ldisk(\vecs)$ that does not contain $\{t,t'\}$. 
  If  $D \in \FIL(P)$, then we choose $F=D$ and we are done. 
  Otherwise, we obtain $F$ by translating $D$ upward until it
  meets the forward chain, and we have $t \in \LDA(F)$.
\end{proof}

The lemma above does not give us a complete characterization of the
reachable region, as we do not know yet whether we can reach the disks
in $\FIL(P)$ and $\LFIL(P)$ tangentially. The following lemma
addresses this issue.

\begin{lemma}\label{lem:reachfilling}
Assume that the forward chain $\FC(\vecs)$ does not intersect the interior of
$\ldisk(\vecs)$. (i) If $F \in \FIL(P)$, then there exists 
a counterclockwise configuration tangent to $F$ that can be reached 
from $\vecs$ by a $CS$ path. (ii) If $F \in \LFIL(P)$, then the first 
configuration on $\FC(\vecs)$ tangent to $F$ can be reached
from $\vecs$ by a $CSC$ path.
\end{lemma}
\begin{proof}
  The lemma is obvious when $F$ touches the edge containing $s$, so we
  assume this is not the case.  We first prove~(i). If we draw a line
  segment upward from the leftmost point~$p$ of~$F$ until we meet the
  forward chain, we do not intersect~$\ldisk(\vecs)$. It follows
  that~$p \in \LDA(\vecs)$.  We consider a ray that starts at a
  configuration~$\vecs'$ tangent to~$\ldisk(\vecs)$. We start at
  $\vecs'=\vecs$ and move $\vecs'$ counterclockwise along the boundary
  of $\ldisk(\vecs)$, so that the ray sweeps~$\LDA(\vecs)$. Since $p
  \in \LDA(\vecs)$, this ray must meet~$F$ at some point.  When the
  ray first meets~$F$, it is tangent to $F$, so we can reach the
  corresponding configuration by a CS~path.

  We now prove~(ii). We denote by $\vecr=(r,\dir)$ the first
  configuration on the forward chain that is tangent to~$F$. Then $r
  \in \LDA(\vecs)$, so when we sweep the same ray as in the proof
  of~(i), we meet~$F$ tangentially at some configuration~$\vecr'$. The
  arc between~$\vecr'$ and~$\vecr$ is inside $P$, so we have a
  CSC~path starting from $\vecs$ and going through $\vecr'$
  and~$\vecr$.
\end{proof}

We are now able to give the following characterization for
the reachable region starting from a configuration on the
side of $P$. It follows directly from Lemmas~\ref{lem:pocketreach},
\ref{l:tinDL}, and~\ref{lem:reachfilling}.
\begin{prop}\label{p:reachside} 
Assume that $\vecs$ is a configuration on the boundary of $P$,
oriented counterclockwise.
Then any point in $\reach(\vecs)$ can be reached by a CSCS path.
In addition, we have that:
\begin{itemize}
\item[(i)] If $\FC(\vecs)$ intersects the interior of $\ldisk(\vecs)$,
  then $\reach(\vecs) = \LDA(\vecs)$.
\item[(ii)] If $\FC(\vecs)$ does not intersects the interior of
  $\ldisk(\vecs)$,  then  $\reach(\vecs)=\bigcup_{F \in \FIL(P) \cup
  \LFIL(\vecs)} \LDA(F)$. 
\end{itemize}
\end{prop}

In the characterization above, it seems that an infinite number of disks
could possibly contribute to the boundary of the reachable region. In the 
following, we show that the contribution of the $\LDA$s along any edge 
can be reduced to at most two $\LDA$s. We focus
on a particular edge $f$. Let $\dir$ denote the counterclockwise direction
along this edge. Then we order the counterclockwise configurations along
$f$ according to direction $\dir$, that is, for two such configurations
$\vecs_1=(s_1,\dir)$ and $\vecs_2=(s_2,\dir)$, we say that 
$\vecs_1 \before \vecs_2$ when
$\overrightarrow{s_1s_2} \cdot \dir \geq 0$.

\begin{lemma}
  \label{l:twoconf}
  Let $\vecs_1=(s_1,\dir)$ and $\vecs_2=(s_2,\dir)$ be two
  counterclockwise configurations on the same edge~$f$ of~$P$, such
  that $\vecs_1 \before \vecs_2$. Let $\vech_1=(h_1,\dir)$ denote the
  first counterclockwise configuration on $f$ such that $\vecs_1 \before
  \vech_1$ and $\ldisk(\vech_1)$ intersects $\FC(\vech_1)\setminus
  \{h_1\}$. Then $\LDA(\vecs_2) \subset \LDA(\vecs_1) \cup
  \LDA(\vech_1)$.
\end{lemma}
\begin{proof}
  We first assume that $\FC(\vecs_1)$ intersects the interior of
  $\ldisk(\vecs_1)$, so that $\vech_1=\vecs_1$. Then $\LDA(s_1)$ can
  be enlarged into a pocket, so by Lemma~\ref{l:pocket}, we have
  $\LDA(\vecs_2)\subset \LDA(\vecs_1)$.
  
  Otherwise, $\FC(\vecs_1)$ does not intersect the interior
  of $\ldisk(\vecs_1)$. Thus, the disk $\ldisk(\vech_1)$ touches
  $\FC(\vecs_1)$. (See Figure~\ref{f:com2}.) 
  \epsfigure{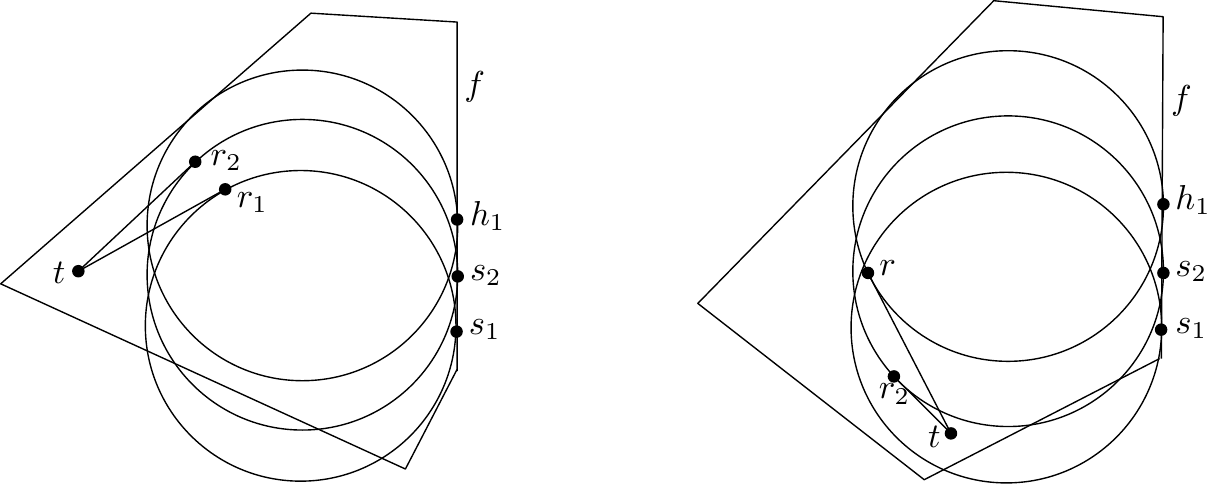}{Proof of Lemma~\ref{l:twoconf}.}
  If $\vech_1 \before \vecs_2$,
  then $\LDA(\vech_1)$ is a pocket, so by Lemma~\ref{l:pocket}, we have
  $\LDA(\vecs_2) \subset \LDA(\vech_1)$. 

  Finally, we assume that $\vecs_1 \before \vecs_2 \before \vech_1$.
  Let $t$ denote a point in $\LDA(\vecs_2)$.  If~$t$ is reached after
  an arc of $\ldisk(\vecs_2)$ with length less than~$\pi$ followed by
  a line segment~$r_2t$, then it is clearly in $\LDA(\vecs_1)$. (See
  Figure~\ref{f:com2}, left).  On the other hand, if~$t$ is reached by
  an arc of $\ldisk(\vecs_2)$ with length at least~$\pi$, followed by
  a segment~$r_2t$, we claim that $t \in \LDA(\vech_1)$. Let $r$ be
  the point of $\ldisk(\vech_1)$ such that the line $rt$ is tangent to
  $\ldisk(\vech_1)$ from the left.  (See Figure~\ref{f:com2}, right.)
  We have to argue that the arc of~$\ldisk(\vech_1)$ from~$h_1$ to~$r$
  lies in~$P$. This follows from the fact that $\ldisk(\vech_1)$ is
  obtained by translating $\ldisk(\vecs_2)$ along~$f$ until it
  touches~$\FC(\vecs_1)$, that the arc of $\ldisk(\vecs_2)$ from $s_2$
  to~$r$ is in~$P$, and that the arc from $h_{1}$ to~$r_{2}$ is
  shorter than the arc from~$s_2$ to~$r$.
\end{proof}

Now we can show how to construct the reachable region from a
configuration on the boundary.
\begin{prop}
  \label{p:computeside}
  Let $\vecs$ be a counterclockwise configuration on the boundary of
  an $n$-sided convex polygon $P$.  Then we can compute in $O(n)$ time
  a set $\BFIL(\vecs)$ of $O(n)$ disks, such that
  $\reach(\vecs)=\bigcup_{F \in \BFIL(\vecs)} \LDA(F)$.
\end{prop}
\begin{proof}
  We use the characterization of $\reach(\vecs)$ from
  Proposition~\ref{p:reachside}.  If $\FC(\vecs)$ intersects the
  interior of $\ldisk(\vecs)$, then we just set
  $\BFIL(\vecs)=\{\ldisk(\vecs)\}$.  So in the remainder of this
  proof, we assume that $\FC(\vecs)$ does not intersect the interior
  of $\ldisk(\vecs)$.

  If $\FIL(P) \neq \emptyset$, then we first construct the
  contribution of the disks in $\FIL(P)$. By
  Observation~\ref{obs:convd}, we only need to find the unit disks
  whose centers are the vertices of the convex hull of the centers of
  the disks in $\FIL(P)$. These disks are tangent to at least two
  edges of $P$, so their centers lie on the medial
  axis~\cite{agss-ltacv-89,csw-fmasp-99}.  of $P$. We compute this
  medial axis in $O(n)$~time using an algorithm by Aggarwal et
  al.~\cite{agss-ltacv-89}, and then check each edge on the medial
  axis to obtain these disks in $O(n)$~time.

  We now observe that if $\ldisk(\vecs) \in \FIL(P)$, then we can set
  $\BFIL(\vecs) = \FIL(\vecs)$ and are done.
  Thus, in the remainder of this proof, we assume that
  $\ldisk(\vecs) \notin \FIL(P)$, and explain how to find the
  contribution of~$\LFIL(P)$.  

  Let $d$ denote the first point on~$\bd P$, starting from $s$ in
  counterclockwise direction, such that $d \in \ldisk(\vecs)$.  (See
  Figure~\ref{f:medial1}(left).)   Let us call a \emph{candidate
    configuration} a configuration~$\vecs'$ on the counterclockwise
  boundary of~$P$ with the property that $\ldisk(\vecs') \subset P
  \cup \ldisk(\vecs)$ and such that $\ldisk(\vecs')$ is either tangent
  to two edges of $\FC(\vecs)$, or is tangent to one edge
  of~$\FC(\vecs)$ and contains the point~$d$.

  Consider an arbitrary edge~$f$ of~$P$.  Let $\vecs_1$ denote the
  first counterclockwise configuration on $f$ 
  such that $\ldisk(\vecs_1) \in \LFIL(P)$. When $\vech_1$ is as in
  Lemma~\ref{l:twoconf}, the $\LDA$s of all the disks in $\LFIL(P)$
  that are tangent to $f$ are contained in $\LDA(\vecs_1) \cup
  \LDA(\vech_1)$. So we only need to find $\vecs_1$ and $\vech_1$ to
  construct the contribution of $f$ to $\reach(\vecs)$.

  We observe that $\vecs_1$ and~$\vech_{1}$ are candidate
  configurations: in fact, $\vecs_1$ is the \emph{first} candidate
  configuration on~$f$, while $\vech_{1}$ is the \emph{last}.

  \epsfigure{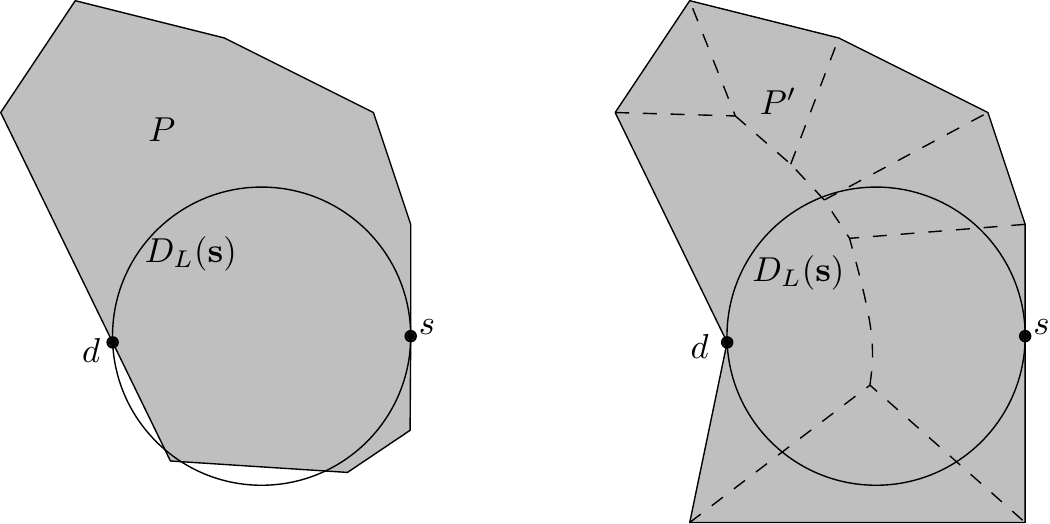}{On the left, $P$ is shaded. On the right, $P'$
  is shaded, and its medial axis is dashed.}

  It remains to explain how to compute the candidate configurations
  efficiently.  Denote by $P'$ a simple polygon obtained by replacing
  the subchain of $\bd P$ that goes counterclockwise from~$d$ to~$s$
  with two or three edges, such that $\ldisk(\vecs) \subset P'$. (See
  Figure~\ref{f:medial1}(right).)  The disk $\ldisk(\vecs')$, for a
  candidate configuration~$\vecs'$, must lie in~$P'$ and must
  touch~$\bd P'$ in more than one point.  It follows that we can find
  the candidate configurations by first computing the medial axis
  of~$P'$ in $O(n)$ time using the algorithm by Chin et
  al.~\cite{csw-fmasp-99}, and then checking all edges of the medial
  axis.
\end{proof}

Proposition~\ref{p:computeside} shows that the reachable region for a
configuration on the boundary of~$P$ is delimited by~$O(n)$
disks. This bound is tight, as shown by the example in
Figure~\ref{f:linear}, where $\Omega(n)$ disks of~$\LFIL(\vecs)$
contribute to the boundary of the reachable region.
\epsfigure{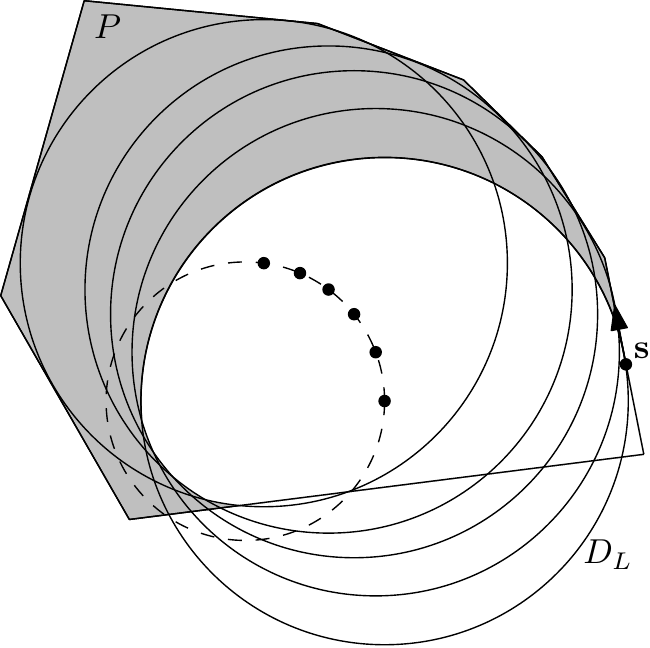}{Example where $\Omega(n)$ disks of  
  $\LFIL(\vecs)$ contribute to the boundary of the reachable
  region. Here the disks of  $\LFIL(\vecs)$ are centered along
  the dashed circle. The reachable region is shaded.}

\section{Special left-right and right-left paths suffice}
\label{s:righleft}

Let $\vecs$ be a starting configuration. A \emph{canonical RL-start}
from~$\vecs$ is a path from~$\vecs$ to a configuration~$\vecr$ on the
boundary of~$P$ that begins with a right-turning arc of unit radius
and continues with a left-turning arc of unit radius ending at~$\vecr$
(and tangent to the boundary of~$P$ there) (Figure~\ref{f:can}).
\epsfigure{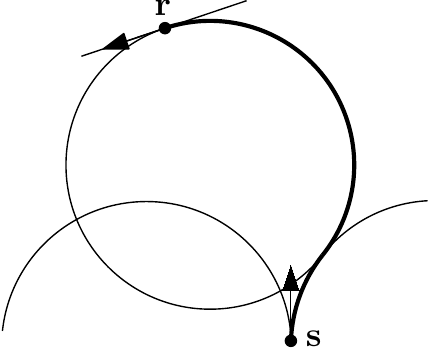}{A canonical RL-start.}

Note that for each edge~$f$ of~$P$ and for a given~$\vecs$, there are
at most two canonical RL-starts from~$\vecs$ ending on~$f$. A
canonical LR-start is defined analogously: it begins with a
left-turning arc and continues with a right-turning arc.

In this section, we show that for determining the reachability by
paths in a convex polygon, it suffices to consider paths of a fairly
special form. Namely, we show that a point is reachable if and only if
it is directly accessible, or it can be reached by a path that begins
with a canonical start. 

Dubins~\cite{d-cmlca-57} showed that the shortest path of
bounded curvature between two configurations in the plane
is of type CSC or CCC. In the latter case, the middle arc
has length more than $\pi$. (See Figure~\ref{f:dubins}.)
We call these paths \emph{Dubins paths}. 

\epsfigure{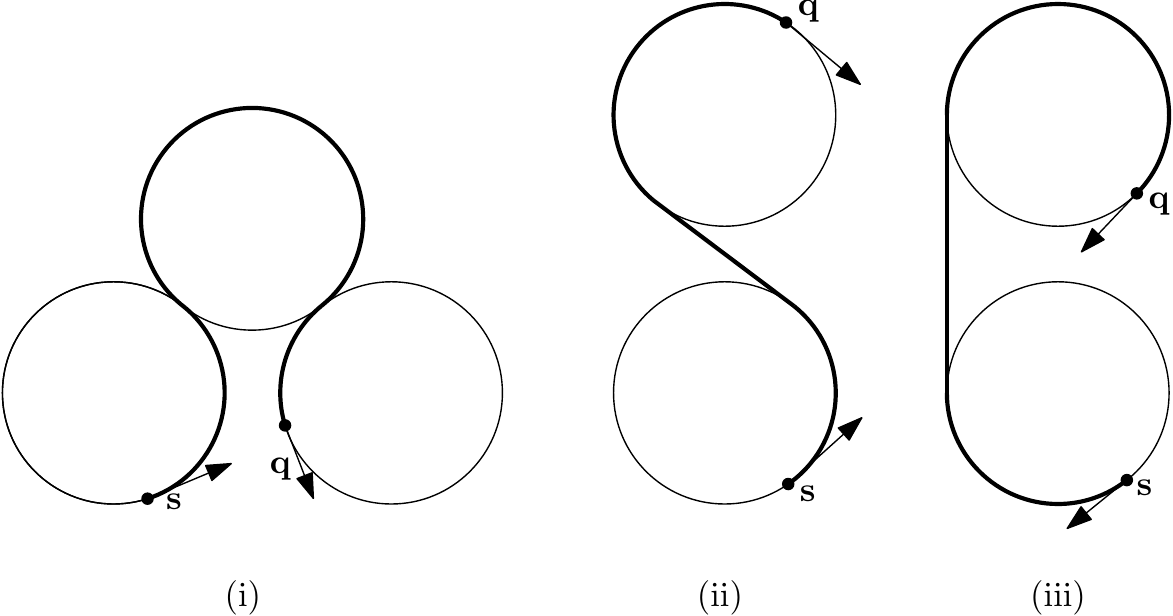}{Three types of Dubins paths.}

\begin{prop}\label{p:rl-enough}
  Let $P$ be a convex polygon, let~$\vecs$ be a starting configuration
  in~$P$, and let $t \in P$ be reachable from~$\vecs$ by a
  bounded-curvature path.  Then $t$ lies in the directly accessible
  region $\DA(\vecs)$, or it can be reached by a path of one of the
  following forms: a canonical RL-start followed by a left-turning
  path (starting on a side of $P$), or a canonical LR-start followed
  by a right-turning path (starting on a side of $P$).
\end{prop}
\begin{proof}
  Jacobs and Canny~\cite{jc-pspmr-92} showed that, in a polygonal
  environment, the shortest path of bounded curvature between two
  configurations is a sequence of Dubins paths.  The final
  configurations of these Dubins paths (except for the last one) all
  lie on the boundary of the polygonal environment. So, if we denote
  by $\gamma$ a shortest path from $\vecs$ to~$t$, then $\gamma$ can
  be written as a sequence $\gamma=\gamma_1 \gamma_2 \dots \gamma_m$
  of $m$~Dubins paths. When $m \geq 2$, we know that the final
  configuration $\vecq$ of $\gamma_1$ lies on $\bd P$.

  We handle three cases separately, according to the type of
  $\gamma_1$ (see Figure~\ref{f:dubins}). If $\ldisk(\vecs)$ and
  $\rdisk(\vecs)$ are contained in $P$, then $\DA(\vecs)=P$, so from
  now on, we assume that $\ldisk(\vecs)$ or $\rdisk(\vecs)$ crosses
  $\bd P$.

  \epsfigure{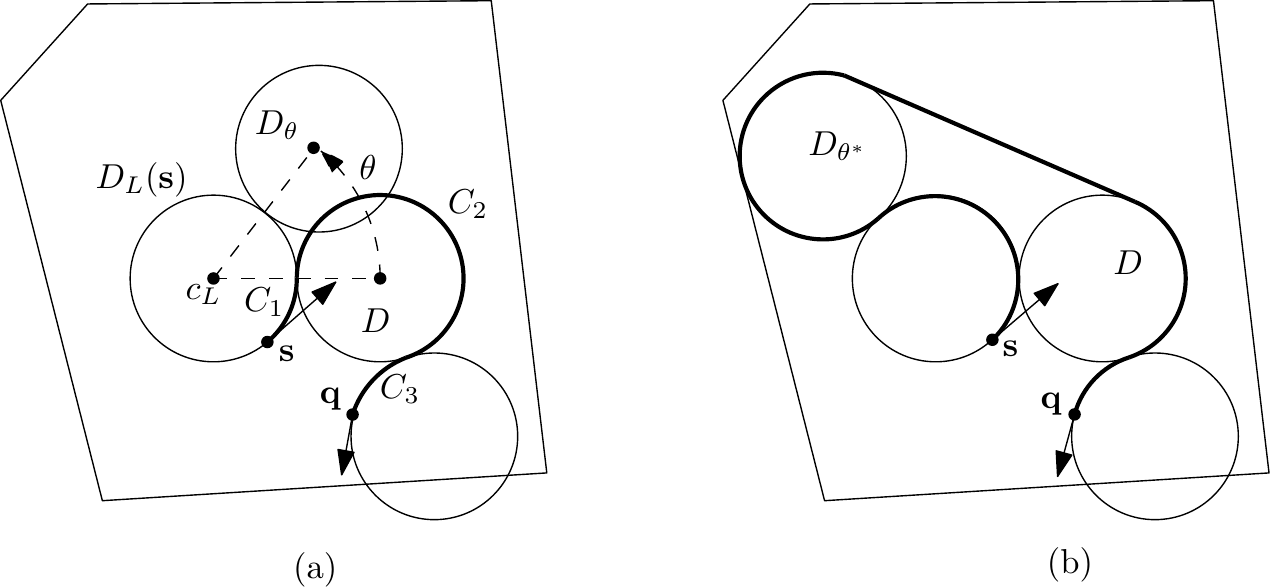}{Proof of Proposition~\ref{p:rl-enough},
  case~(i).}
  \paragraph{Case~(i).} We assume that $\gamma_1$ is of type $CCC$.
  We denote by $C_1$, $C_2$, $C_3$ the three circle arcs such that
  $\gamma_1=C_1C_2C_3$, and recall that $C_2$ has length larger
  than~$\pi$.  If $C_2$ touches $\bd{P}$ we are done, so from now on
  we assume that $C_2$ does not touch $\bd{P}$.  Let $D$ be the disk
  supporting~$C_2$.  Without loss of generality, we assume that $C_1$
  and $C_3$ turn counterclockwise and $C_2$ turns clockwise.  (See
  Figure~\ref{f:casei}(a).)  Let $c_L$ denote the center
  of~$\ldisk(\vecs)$.  We denote by $D_\theta$ the disk obtained by
  rotating~$D$ by an angle~$\theta$ around~$c_L$. We define
  $\thetastar$ as the smallest~$\theta \in (0,2\pi]$ such that $\theta
    \leq \pi/3$ and $D_\theta \setminus D$ touches $\bd{P}$, or
    $\theta > \pi/3$ and $D_\theta$ touches $\bd{P}$.  Since
    $\ldisk(\vecs) \cup \rdisk(\vecs)\not\subset P$, such a
    $\thetastar$ exists.  There is a path from $\vecs$ to $\vecq$
    consisting of an arc of~$\ldisk(\vecs)$, an arc of $D_{\theta^*}$,
    a line segment, an arc of~$D$ and $C_3$. (See
    Figure~\ref{f:casei}(b).)  This means that $t$ can be reached by a
    canonical LR-start and a right-turning path.

  \epsfigure{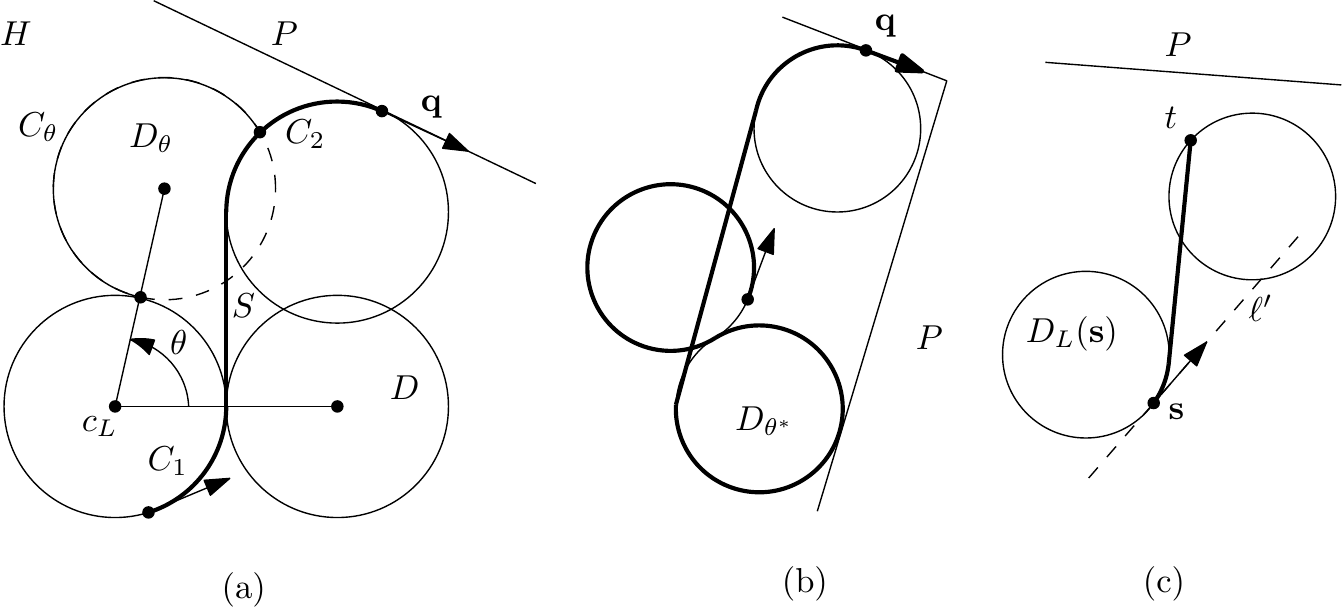}{Proof of Proposition~\ref{p:rl-enough},
  case~(ii).}
  \paragraph{Case~(ii).} We assume that $\gamma_1=C_1SC_2$,
  where $C_1$ is left-turning, $S$ is a segment and 
  $C_2$ is right-turning. (See Figure~\ref{f:caseii}.)

  Let us first assume that $m = 1$, that is that $\gamma = \gamma_1 =
  C_1SC_2$ and that $C_2$ has length less than~$\pi$.  In this case,
  $t$ lies in~$\DA(\vecs)$.  Indeed, if $C_1$ has length larger
  than~$\pi$ or if $t$ lies to the left of the directed line~$\ell'$
  defined by~$\vecs$, then $t \in \LDA(\vecs)$. (See
  Figure~\ref{f:caseii}(c)). If $C_1$ has length less than~$\pi$ and
  $t$ lies to the right of~$\ell'$, then $t \in \RDA(\vecs)$, since
  $\gamma$ cannot enter $\rdisk(\vecs)$ because $C_2$ has length less
  than~$\pi$.

  It remains to consider the case where $m \geq 2$ or the length of
  $C_2$ is at least~$\pi$. Let $D$ be the unit disk tangent to
  $\ldisk(\vecs)$ and lying to the right of the segment~$S$, and
  denote by $D_\theta$ the disk obtained by rotating~$D$ by an
  angle~$\theta \in (0,2\pi]$ counterclockwise around~$c_L$.  (See
    Figure~\ref{f:caseii}(a).) Let $C_\theta$ be the clockwise arc
    of~$D_\theta$ starting at $D_L(\vecs) \cap D_\theta$ and going clockwise
    until the first intersection point with the path~$SC_2$, or
    returning to its starting point if there is no such intersection.
    Let $\thetastar$ be the smallest value of $\theta \in (0,2\pi]$
      such that $C_\theta$ touches~$\bd P$. Since $\ldisk(\vecs) \cup
      \rdisk(\vecs)\not\subset P$, such a $\thetastar$ exists.  Now
      $\vecq$ can be reached by a path consisting of an arc of
      $\ldisk(\vecs)$, an arc of $D_{\thetastar}$, a line segment, and
      an portion of~$C_2$. (See Figure~\ref{f:caseii}(b).)  It follows
      that $t$~can be reached by a canonical LR-start and a
      right-turning path.

  \epsfigure{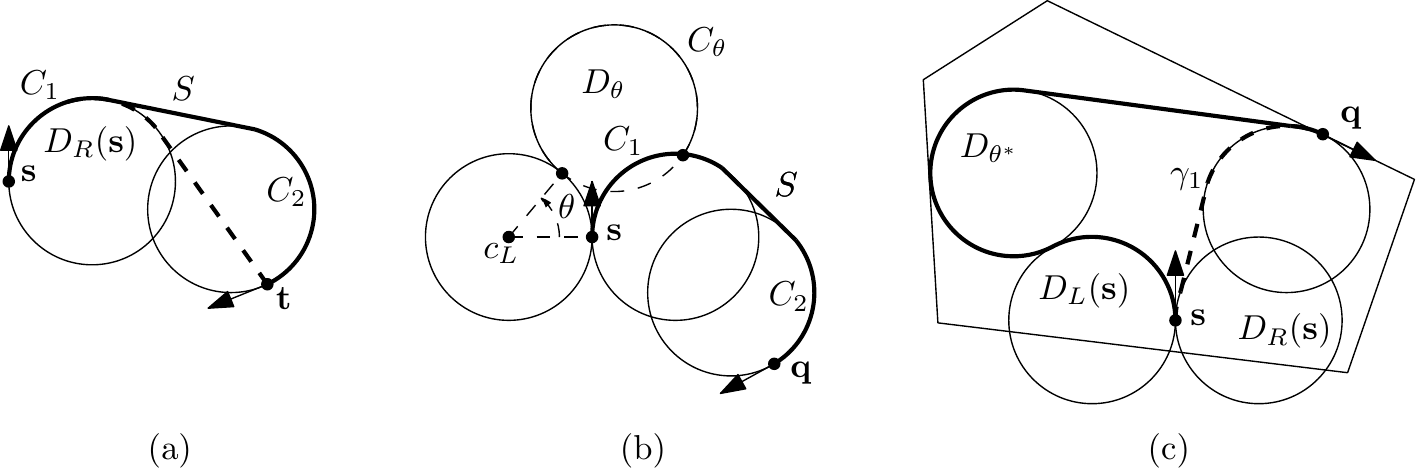}{Proof of Proposition~\ref{p:rl-enough},
  case~(iii).}
  \paragraph{Case~(iii).} We assume that $\gamma_1=C_1SC_2$,
  where $C_1$ and $C_2$ are right-turning and $S$ is a segment.  If
  the length of $C_2$ is less than~$\pi$ and $m=1$ (see
  Figure~\ref{f:caseiii}(a)), then $t \not \in \rdisk(\vecs)$, and
  therefore $t \in \RDA(\vecs)$.  We therefore assume that $m \geq 2$
  or the length of $C_2$ is at least~$\pi$ .

  We denote by $D_\theta$ the disk obtained by rotating
  $\rdisk(\vecs)$ around the center~$c_L$ of~$\ldisk(\vecs)$ by an
  angle $\theta \in (0,\pi/2]$. (See Figure~\ref{f:caseiii}(b).) When
    $D_\theta$ intersects~$\gamma_1$, we denote by $C_\theta$ the arc
    of $\bd D_\theta$ that starts at $\ldisk(\vecs) \cap D_\theta$ and
    goes clockwise until it meets~$\gamma_1$. Otherwise, we denote
    $C_\theta=\bd D_\theta$.  As before, let $\thetastar$ be the
    smallest value of $\theta \in (0,2\pi]$ such that $C_\theta$
      touches $\bd P$. Again, $\thetastar$ exists since $\ldisk(\vecs)
      \cup \rdisk(\vecs)\not\subset P$.  Now $\vecq$ can be reached by
      a path that consists of an arc of~$\ldisk(\vecs)$, an arc of
      $D_{\thetastar}$, a segment tangent to $D_{\thetastar}$
      and~$\gamma_1$, and a subpath of~$\gamma_1$. (See
      Figure~\ref{f:caseiii}(c).)  This implies again that~$t$ can be
      reached by a canonical LR-start followed by a right-turning
      path. 
\end{proof}

We can give a somewhat different characterization:
\begin{prop} 
  Let $P$ be a convex polygon, let $\vecs$ be a starting configuration
  in~$P$, and let $t \in P$ be reachable from $\vecs$ by a
  bounded-curvature path. Then $t$ is reachable by a path of the
  form~CCSCS. More precisely, $t$ is reachable by a path of the form
  CS, or it is reachable by a path of the form CCSCS, where the two
  final disks touch the boundary of~$P$, and the path goes through
  these touching points.
\end{prop}

\section{Putting everything together}
\label{s:together}

In this section, we show how to construct the reachable region
when~$\vecs$ is an arbitrary configuration in~$P$. We obtain it by
combining the results in sections~\ref{s:side} and~\ref{s:righleft}.
We will prove the following:
\begin{theorem}\label{th:main}
  Let $P$ be an $n$-sided, convex polygon, and let $\vecs$ be a
  configuration inside~$P$. Then the reachable region $\reach(\vecs)$
  from~$\vecs$ inside~$P$ is delimited by $O(n)$ arcs of unit circles,
  and we can compute $\reach(\vecs)$ in $O(n^2)$~time.
\end{theorem}
\begin{proof}
  Let $t$ be a point in $\reach(\vecs)$. By
  Proposition~\ref{p:rl-enough}, either $t$ is in $\DA(\vecs)$, or it
  can be reached after a canonical start.  We only consider canonical
  RL-starts; the case of LR-starts can be handled symmetrically.

  The directly accessible region $\DA(\vecs)$ is delimited by two
  circle arcs, which can be computed in $O(n)$ time by brute force. We
  determine the at most $2n$~canonical RL-starts by brute force, in
  $O(n^2)$ time. For each canonical start, by
  Proposition~\ref{p:computeside}, we compute in $O(n)$ time a set of
  $O(n)$ configurations on the side of $P$ such that the union of
  their $\LDA$s form the reachable region after this canonical start.

  We have thus obtained a set of $O(n^2)$~configuration on the
  boundary of~$P$ such that the union of their $\LDA$s with
  $\DA(\vecs)$ is $\reach(\vecs)$.  By Lemma~\ref{l:twoconf}, we only
  need to keep two such configurations per edge: the first and the
  last one. As we have only $O(n)$ arcs to consider, we can construct
  $\reach(\vecs)$ by inserting these arcs one by one, and updating the
  reachable region by brute force.  As these arcs are arcs of unit
  circles, each one of them appears only once along the boundary of
  $\reach(\vecs)$. So overall, it takes $O(n^2)$ time.
\end{proof}

\section*{Acknowledgments}

This problem was first posed to us by Hazel Everett.  We miss her. We
also thank Sylvain Lazard, Ngoc-Minh L\^e, and Steve Wismath for
discussions on this problem.


\end{document}